\documentclass{article}

\usepackage{amsmath,amsfonts,bm}
\def\eqref#1{equation~\ref{#1}}

\def\1{\bm{1}}

\DeclareMathAlphabet{\mathsfit}{\encodingdefault}{\sfdefault}{m}{sl}
\SetMathAlphabet{\mathsfit}{bold}{\encodingdefault}{\sfdefault}{bx}{n}

\DeclareMathOperator*{\argmin}{arg\,min}

\usepackage{microtype}
\usepackage{graphicx}
\usepackage{subfigure}
\usepackage{booktabs} 
\usepackage{multirow}
\usepackage{amsmath}
\usepackage{amssymb}
\usepackage{amsthm}
\usepackage{xcolor}
\usepackage{float}
\usepackage{hyperref}

\newtheorem{proposition}{Proposition}

\newcommand{\lf}[1]{{\color{green}LIAM---#1---}}

\newcommand{\jg}[1]{{\color{teal}JONAS---#1---}}

\usepackage[accepted]{icml2021}

\icmltitlerunning{Preventing Unauthorized Use of Proprietary Data: Poisoning for Secure Dataset Release}

\begin{document}

\twocolumn[
\icmltitle{Preventing Unauthorized Use of Proprietary Data: \\ Poisoning for Secure Dataset Release}



\icmlsetsymbol{equal}{*}

\begin{icmlauthorlist}
\icmlauthor{Liam Fowl}{equal,math}
\icmlauthor{Ping-yeh Chiang}{equal,cs}
\icmlauthor{Micah Goldblum}{equal,cs}
\icmlauthor{Jonas Geping}{germany}
\icmlauthor{Arpit Bansal}{cs}
\icmlauthor{Wojtek Czaja}{math}
\icmlauthor{Tom Goldstein}{cs}
\end{icmlauthorlist}

\icmlaffiliation{math}{Department of Mathematics, University of Maryland}
\icmlaffiliation{cs}{Department of Computer Science, University of Maryland}
\icmlaffiliation{germany}{Dep. of Electrical Engineering and Computer Science, University of Siegen}

\icmlcorrespondingauthor{Liam Fowl}{lfowl@umd.edu}

\icmlkeywords{Machine Learning, ICML}

\vskip 0.3in
]



\printAffiliationsAndNotice{\icmlEqualContribution} 

\begin{abstract}
%
%
Large organizations such as social media companies continually release data, for example user images.  At the same time, these organizations leverage their massive corpora of released data to train proprietary models that give them an edge over their competitors.  These two behaviors can be in conflict as an organization wants to prevent competitors from using their own data to replicate the performance of their proprietary models.  We solve this problem by developing a data poisoning method by which publicly released data can be minimally modified to prevent others from training models on it.  Moreover, our method can be used in an online fashion so that companies can protect their data in real time as they release it.  We demonstrate the success of our approach on ImageNet classification and on facial recognition. 
\end{abstract}

\section{Introduction}
\label{intro}
 Social media companies and other web platforms allow users to post their own data in a space that is openly accessible to scraping \cite{taigman_deepface:_2014, cherepanova2021lowkey}. This data can be incredibly valuable, both to the organization hosting it and to others who leverage scraped data to train their own models. Breakthroughs in both image classification \cite{russakovsky_imagenet_2015} and language models \cite{brown2020language} have been enabled by large volumes of scraped data. Given that organizations value their exclusive access to the data they host for training competitive machine learning systems, they need a method for safely releasing their data on their web platform while preventing competitors from replicating the performance of their own models trained on this data. 

We introduce a method, motivated by new techniques in targeted data poisoning, to modify data prior to its release so that the generalization of a deep learning model trained on this data is significantly degraded, rendering the data effectively worthless to competitors. 

\subsection{Related Work}
\label{rel_work}

The topic of data manipulation for the purposes of performance degradation has been investigated in the \textit{data poisoning} literature. Specifically, this work is closely related to \textit{indiscriminate} (availability) poisoning attacks wherein an attacker wishes to degrade performance on a large number of samples \cite{barreno_security_2010}. Early works on this type of attack show that data can be maliciously modified to degrade test-time performance of simple classical algorithms, such as support vector machines, principle component analysis, clustering, logistic regression, etc., or in the setting of binary classification. \cite{munoz-gonzalez_towards_2017, xiao_is_2015, biggio_poisoning_2012, koh_stronger_2018, steinhardt_certified_2017}. \par 
In scenarios involving simple learning models, the optimal perturbation to training data can often be explicitly calculated via the implicit function theorem. However, this becomes computationally intractable for modern deep networks. As a consequence, little work has been done on indiscriminate attacks on deep networks. Recently, \citet{shen_tensorclog:_2019} proposed a heuristic to avoid having to explicitly solve the full bi-level objective. The method, TensorClog, crafts perturbations to cause gradient vanishing with the aim of preventing a deep network from training on the perturbed data, thus degrading test time performance of the network. However, this work only performs their method in the setting of transfer learning where a known feature extractor is used, limiting the viability of this attack. 

In contrast to indiscriminate attacks, \textit{targeted} (integrity) poisoning attacks aim to cause a network trained on modified data to mis-classify a few pre-selected target samples. Unlike the indiscriminate attack setting, recent work on targeted poisoning has successfully attacked modern deep networks trained from scratch on poisoned data \cite{geiping2020witches, huang_metapoison:_2020}. These attacks do not noticeably degrade validation accuracy, despite the victim network mis-classifying the selected target example(s). Moreover, simply performing a large number of targeted attacks to degrade overall validation accuracy is not feasible. In some cases, these attacks perturb up to $10\%$ of the training data in order to mis-classify a \textit{single} target image, and they often fail to attack more than a handful of targets \cite{geiping2020witches}. These attacks rely upon the expressiveness of deep networks to ``gerrymander" the decision boundary of the victim network around the selected target examples - a strategy that will not work in the general indiscriminate setting.

We develop an indiscriminate data poisoning attack which works on deep networks trained from scratch in a black-box setting. Our method allows practitioners to minimally modify data which, when released, causes models trained on this data to generalize poorly. Our method allows companies to release data, either for transparency purposes, or via user upload, which does not compromise the competitive advantage the company gains from asymmetric access to the clean data. For a general overview of data poisoning attacks, defenses, and terminology, see \citet{goldblum_dataset_2020}. 

Also parallel to the goals of this work are defenses to model stealing attacks. Model stealing attacks often aim duplicate a machine learning model or its functionality \cite{tramer2016stealing}. Defenses vary depending upon the attack scenario. For example, a defense proposed in
 \citet{orekondy2019prediction} perturbs the prediction outputs of a network to prevent a model stealing attack which aims to mimic the performance of a service like a cloud prediction API. Other defenses aim to prove theft of a model has occurred after the fact by ``watermarking" a network \cite{uchida2017embedding}. However, these defenses do little to stop malicious actors from using scraped data to train their own models \cite{clearviewai}. 
\section{Our Method}
\label{method}

\subsection{Problem Setup}
\label{setup}
Formally, we seek to compute perturbations $\Delta = \{ \Delta_i \}$ to elements $x_i$ of a dataset $\mathcal{S}$ in order to make a network, $F$, trained on the dataset generalize poorly to the distribution $\mathcal{D}$ from which $\mathcal{S}$ was sampled. Achieving this goal entails solving the following bi-level objective, 

 \begin{gather}\label{eq:bilevel}
     \max_{\Delta \in \mathcal{C}} \,\, \mathbb{E}_{(x,y) \sim \mathcal{D}} \bigg[ \mathcal{L} \left( F(x; \theta(\Delta)), y \right) \bigg] \\
     \text{s.t.} \,\, \theta(\Delta) \in \argmin_\theta\sum_{(x_i, y_i) \in \mathcal{S}} \mathcal{L}(F(x_i + \Delta_i; \theta), y_i),
\end{gather}
 where $\mathcal{C}$ denotes the constraint set which bounds the perturbations so that the perturbed data is perceptually similar to the clean data. In our work, we employ the $\ell_{\infty}$ constraint $\Vert\Delta\Vert_\infty < \epsilon$ as is standard in the adversarial literature \cite{madry_towards_2017, zhu_transferable_2019, geiping2020witches}. Constraining the perturbations in this fashion allows practitioners like social media companies to employ our method in order to release minimally changed user data while still protecting the performance of their proprietary models. 
 
 Directly solving for $\Delta$ which minimizes this objective is intractable as this would require backpropagating through the entire training procedure found in the inner objective (2) for each iteration of gradient descent on the outer objective. Thus, the bilevel objective must be approximated. 

\begin{algorithm}[bht!]
\caption{Crafting perturbations}
\label{alg:crafting}
\begin{algorithmic}[1]
    \STATE {\bfseries Require} pre-trained clean network $\{F(\cdot, \theta)\}$, 
    a dataset $\mathcal{S} = \{(x_i, y_i)\}_{i=1}^N$, perturbation bound $\varepsilon$,  restarts $R$, optimization steps $M$ 
    \STATE {\bfseries Begin}
    Compute $ \nabla_\theta \mathcal{L}'(\mathcal{S}; \theta)$ i.e. the target gradient.
    \STATE {\bfseries For} $r = 1, \dots, R$ restarts:
    \STATE \quad Randomly initialize perturbations $\Delta^r \in \mathcal{C}$
    \STATE \quad {\bfseries For} $j = 1, \dots, M$ optimization steps:
    \STATE \quad \quad Apply data augmentation to samples $(x_i+\Delta^r_i)_{i=1}^N$
    \STATE \quad \quad Compute alignment loss, $\mathcal{A}(\Delta^r, \theta)$ as in  Eq. \ref{eq:alignment_loss} 
    \STATE \quad \quad Update $\Delta^r$ with a step of signed Adam 
    \STATE \quad \quad Project onto $||\Delta^r||_\infty \leq \varepsilon$
    \STATE Choose $\Delta^*$ as $\Delta^r$ with minimal value in $\mathcal{A}(\Delta^r,\theta)$ 
    \STATE {\bfseries Return} perturbations $\Delta^*$
\end{algorithmic}
\end{algorithm}


\subsection{Crafting Perturbations}
It has been demonstrated in Witches' Brew \cite{geiping2020witches} that bounded perturbations to training data can be crafted to manipulate the gradient of a network trained on this data. We adapt gradient manipulation to the problem of general performance degradation. We estimate the outer objective $(1)$ by training a network $F$ on a clean dataset $\mathcal{S}$ and then crafting perturbations to minimize the following objective: 

\begin{equation}
    \mathcal{A}(\Delta, \theta) = 1-~\frac{\big \langle \nabla_\theta \mathcal{L}'(\mathcal{S}; \theta), \nabla_\theta \mathcal{L}(\mathcal{S} + \Delta; \theta) \big \rangle}{\Vert \nabla_\theta \mathcal{L}'(\mathcal{S}; \theta)\Vert_2 \cdot \Vert \nabla_\theta \mathcal{L}(\mathcal{S} + \Delta; \theta) \Vert_2 }
\label{eq:alignment_loss}
\end{equation}
where $\mathcal{L}'(F(x_i; \theta), y_i) $ is the ``reverse" cross entropy loss \cite{huang2019understanding} which discourages the network $F$ from classifying $x_i$ with label $y_i$. Specifically, the reverse cross entropy loss for a sample $(x,y)$ with one-hot label $y$ is given by:
\[\mathcal{L}'(F(x; \theta), y) = -\log[1-p_\theta(x)_{y \neq 0}]\]
where $p_\theta(x)_{y \neq 0}$ denotes the entry of the softmax of $F(x;\theta)$ corresponding to the class specified by the one-hot label $y$.
For notational simplicity, we denote the \textit{target gradient}
\begin{equation}
    \label{eq:targ_grad}
    \nabla_\theta \mathcal{L}'(\mathcal{S}; \theta) = \nabla_\theta \sum_{(x_i,y_i) \in \mathcal{S}}\mathcal{L}'(F(x_i; \theta), y_i) 
\end{equation}
and the \textit{crafting gradient}
\[ \nabla_\theta \mathcal{L}(\mathcal{S} + \Delta; \theta) = \nabla_\theta \left( \sum_{(x_i,y_i) \in \mathcal{S}}\mathcal{L}(F(x_i + \Delta_i; \theta), y_i) \right). \]

Put simply, we seek to align the training gradient of the perturbed data with the gradient of the reverse cross-entropy loss on the clean data, $\mathcal{S}$. Ideally, when a network trains on the perturbed data, the perturbations cause the training gradient of this network at each parameter vector to be aligned with the gradient of the reverse cross entropy loss on the clean data. This in turn would decrease the reverse cross entropy loss on the clean data, causing a network trained on the perturbed data to converge to a minimum with poor generalization on the distribution from which $\mathcal{S}$ was sampled. 

In order to enforce the constraints in $\mathcal{C}$, we employ projected gradient descent (PGD) as in \citet{madry_towards_2017} on the perturbations, alternately minimizing Eq. \ref{eq:alignment_loss} and projecting onto an $l_\infty$ ball. Algorithm \ref{alg:crafting} details this procedure. 

Additionally, we employ techniques found in previous poisoning work such as restarts and differentiable data augmentation in the crafting procedure in order to improve the success of our perturbations \cite{huang_metapoison:_2020,geiping2020witches}. 

We pre-train a model in order to estimate the target gradient, and the crafting gradients as already trained models have been shown to provide the most stable perturbations in \citet{geiping2020witches, huang_metapoison:_2020}. Additional details about the training procedure and the crafting procedure can be found in the appendix \ref{ap:training}. 

\begin{figure*}[thb]
    \centering
    \includegraphics[scale=0.7]{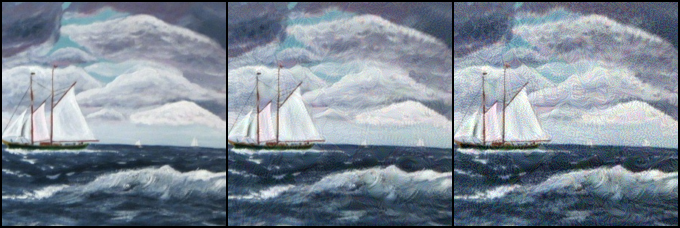}
    \caption{Randomly selected example perturbations to ImageNet datapoint (class ``schooner"). \textbf{Left}: unaltered base image. \textbf{Middle}: $\varepsilon = 8/255$ perturbation. \textbf{Right}: $\varepsilon=16/255$ perturbation.}
    \label{fig:imagenet_grid}
\end{figure*}

\section{Experimental Results}

\label{results}
\subsection{Setup}
We establish baselines for our method on both the ILSVRC2012 dataset (ImageNet) \cite{russakovsky_imagenet_2015} and the CIFAR-10 dataset \cite{krizhevsky2009learning}. ImageNet consists of over $1$ million images coming from $1000$ classes. For these experiments, to save time crafting, we use a pre-trained model from \texttt{torchvision} (see  \url{https://pytorch.org/docs/stable/torchvision/models.html}). For memory constraints, we then craft poisons in batches of $25,000$, replacing the target gradient estimate each split. We train ImageNet models for $40$ epochs.

CIFAR-10 consists of $50,000$ images coming from $10$ classes. For the baseline experiments, we calculate the target gradient on the entire training set. We train CIFAR-10 models for $50$ epochs. 

Unless otherwise stated, we craft our poisons by choosing ResNet-18 \cite{he_deep_2015} as the architecture for $F$ using $8$ restarts and $240$ optimization steps using signed Adam as in Witches' Brew. 

For evaluation, we train a new, randomly initialized network from scratch. We test our method both on the network which was used for crafting, and in a black-box setting where the network architecture is unknown to the practitioner. See subsection \ref{subsec:black_box} for more details on evaluation.

\subsection{ImageNet}
To test our method on an industrial-scale dataset, we first craft perturbations to  ImageNet. In this setting, we deploy a full crafting procedure which mimics the scenario where a company already has a large corpus of data they wish to release. In this case, the company can train the clean model $F$ on this data, and estimate the average target gradient over this training set. 

\par 
 We find that we are able to significantly degrade the validation accuracy, and in the case of the largest perturbation, decrease it by more than $58\%$ (see Table \ref{tab:imagenet_baseline}).  Visualizations for these perturbations can be found in Figure \ref{fig:imagenet_grid}.

\begin{table}[thb]
\caption{Comparison of validation accuracies of a ResNet-18 \cite{he_deep_2015} trained on perturbed data crafted with different  $\varepsilon$-bounds (using ResNet-18) on ImageNet.}
\label{tab:imagenet_baseline}
\vskip 0.15in
\begin{center}
\begin{small}
\begin{sc}
\begin{tabular}{lc}
\toprule
$\varepsilon$-bound & Validation Acc. (\%) $\downarrow$ \\
\midrule
$0/255$ (clean)    & $65.70$  \\
$8/255$ & $37.58$ \\
$16/255$ & $ 27.49$ \\

\bottomrule
\end{tabular}
\end{sc}
\end{small}
\end{center}

\end{table}

\subsection{Baseline Comparison}
\label{baseline}

Additionally, we establish baseline comparisons for our method on CIFAR-10 by comparing our method to other poisoning methods, data manipulations, and to performance on clean non-poisoned data. We compare our method to the following alternatives:

\textbf{TensorClog} \cite{shen_tensorclog:_2019}: We compare our method to TensorClog, which, to our knowledge, is the only previously existing indiscriminate poisoning attack which has been shown to work on modern deep networks. TensorClog was designed primarily for white-box transfer learning based attacks where a known feature extractor is frozen for evaluation. However, to produce a fair comparison, we re-implement their objective into our crafting regime. This objective aims to cause vanishing training gradients to prevent a network from training properly on the dataset. 

\textbf{Random Noise}: In order to tease apart the effect of crafted versus non-crafted perturbations on validation accuracy, we also compare our method to fixed, random additive noise. Since our PGD based perturbation usually results in a perturbation close to the ``corners" of the $\ell_\infty$ ball, i.e. most pixels are perturbed by the maximum allowed value, we enforce that the noise is of the maximum allowable level for our constraint set. 

Comparisons are made with the same $\varepsilon$-bound and the same training procedure from a randomly initialized ResNet-18 model. We find that our alignment method significantly outperforms these alternatives with the same $\varepsilon$-bound. In the case of ResNet18, we degrade validation accuracy by close $50$ more percentage points than TensorClog, the next best.  

\begin{table}[t]
\caption{Comparison of different poisoning methods. All methods were employed with $\varepsilon$-bound $8/255$ using the same ResNet-18 architecture and training procedure on CIFAR-10. Confidence intervals are of radius one standard error.}
\label{tab:baseline_comparison}
\vskip 0.15in
\begin{center}
\begin{small}
\begin{sc}
\begin{tabular}{cc}
\toprule
Method & Validation Acc. (\%) $\downarrow$ \\
\midrule
None & 93.16 $\pm$ 0.08 \\
Tensorclog & 84.24 $\pm$ 0.17 \\
Random Noise & 90.52 $\pm$ 0.08  \\
Alignment (Ours) & \textbf{36.83 $\pm$ 1.94} \\
\bottomrule
\end{tabular}
\end{sc}
\end{small}
\end{center}
\vspace{-10pt}
\end{table}

\begin{figure*}[bht!]
    \centering
    \includegraphics[scale=0.45]{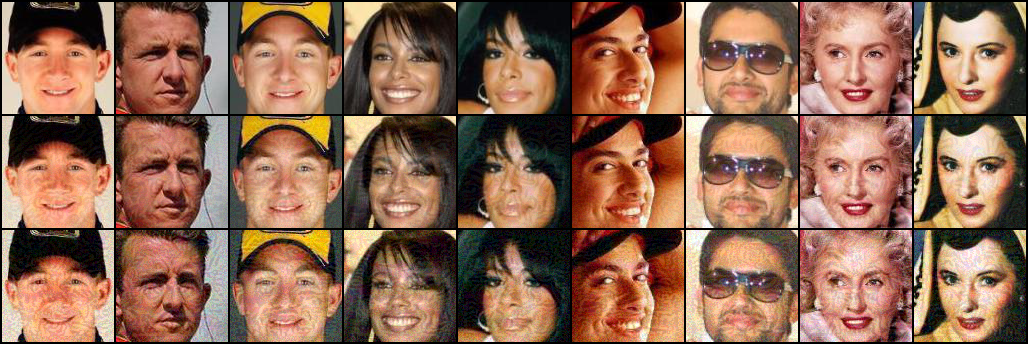}
    \caption{Samples of poisoned CelebA images. \textbf{Top}: unaltered images. \textbf{Middle}: $\varepsilon = 8/255$. \textbf{Bottom}: $\varepsilon=16/255$.}
    \label{fig:celeba_grid}
\end{figure*}

\subsection{Results in the Black-box Setting}
\label{subsec:black_box}
The results in Tables \ref{tab:imagenet_baseline} and \ref{tab:baseline_comparison} are in a grey-box setting where we do not know the victim's model initialization or any information about batching, but we do know their architecture training details like optimizer choice. In order to more strenuously test our method in a realistic setting, we also test our poisons on architectures and training procedures completely unknown during poison crafting. 
This mimics the black-box setting wherein the practitioner does not know how a victim plans to train on their scraped data. Specifically, we train networks on our crafted poisons using the setup from an independent, well known repository for training CIFAR-10 models \footnote{Training routine taken from widely used repository: https://github.com/kuangliu/pytorch-cifar.}. \par 
We validate our method on commonly used architectures including VGG19 \cite{simonyan_very_2014}, ResNet-18 \cite{he_deep_2015}, GoogLeNet \cite{szegedy_going_2015}, DenseNet121 \cite{huang_densely_2016}, and MobileNetV2 \cite{sandler_mobilenetv2:_2018}. These results are presented in Table \ref{tab:epsilon_comparison_full}. We see that even at very small epsilon constraints, the proposed method is able to nearly halve the clean validation accuracy of many of these popular models. 

\begin{table*}[h!]
\caption{Comparison of CIFAR-10 validation accuracies (\%) for different $\varepsilon$-bounds and victim networks. All poisons crafted using a ResNet-18 architecture.}
\label{tab:epsilon_comparison_full}
\vskip 0.15in
\begin{center}
\begin{small}
\begin{sc}
\begin{tabular}{rcccc}
\toprule
Network & $0/255$ (clean) & $4/255$ & $8/255$ & $16/255$  \\
\midrule
VGG19 & 90.76 $\pm$ 0.14 & 66.58 $\pm$ 0.58 & 48.27 $\pm$ 0.92 & 31.86 $\pm$ 1.47 \\
ResNet-18 & 93.16 $\pm$ 0.08 & 68.15 $\pm$ 0.55 & 55.34 $\pm$ 0.80 & 43.14 $\pm$ 1.60 \\
GoogLeNet & 93.87 $\pm$ 0.07 & 71.02 $\pm$ 0.15 & 55.49 $\pm$ 0.58 & 45.37 $\pm$ 0.94 \\
DenseNet121 & 93.80 $\pm$ 0.05 & 70.12 $\pm$ 0.12 & 49.23 $\pm$ 1.58 & 38.51 $\pm$ 1.07 \\
MobileNetV2 &  91.10 $\pm$ 0.10 & 66.71 $\pm$ 0.60 & 51.03 $\pm$ 0.86 & 39.19 $\pm$ 1.76 \\
\bottomrule
\end{tabular}
\end{sc}
\end{small}
\end{center}
\end{table*}

\begin{table}[t]
\caption{Comparison of CIFAR-10 validation accuracies (\%) for different $\varepsilon$-bounds and victim networks. All poisons crafted using a ResNet-18 architecture with estimated target grad.}
\label{tab:epsilon_comparison_partial}
\vskip 0.15in
\begin{center}
\begin{small}
\begin{sc}
\begin{tabular}{rccc}
\toprule
Network & $4/255$ & $8/255$ & $16/255$  \\
\midrule
VGG19 & 74.6 & 56.57 & 30.02 \\
ResNet-18 & 74.11 & 56.65 & 29.62 \\
GoogLeNet & 76.41 & 62.44 & 33.68 \\
DenseNet121 & 75.61 & 57.29 & 31.51 \\
MobileNetV2 &  71.54 & 49.14 & 24.45 \\
\bottomrule
\end{tabular}
\end{sc}
\end{small}
\end{center}

\end{table}

\subsection{Stability}
\label{subsec:stability}
While our method is able to degrade validation accuracy under normal training conditions, a question remains whether the method is stable to poisoning defenses and modifications in training procedures. Would this improve the validation accuracy of a model trained on the perturbed data? To this end, we investigate several avenues. 

First, we investigate whether existing poisoning defenses lessen the effects of our perturbations. Many existing defenses are designed for settings which are significantly different than our proposed attack. For example, many defenses assume that there is a small amount of poisoned data, and that this will be anomalous in feature space \cite{steinhardt_certified_2017-1,tran_spectral_2018, peri_deep_2020, chen2018detecting}. These types of defenses are best suited for scenarios in which the same perturbations are applied identically to each data point, as in patch based attacks, or when the defender has access to a large corpus of trustworthy data on which to train a feature extractor to filter out poisoned images based on their similarity to clean reference samples. Moreover, these anomaly detection defenses work under the assumption that the majority of data will not be poisoned. It was demonstrated in \cite{geiping2020witches} that modifications meant to alter gradients of a network trained from scratch do not produce data which is anomalous in feature space of the poisoned model. Furthermore, since we poison the entire dataset, the trained model's feature space can no longer be thought of as containing clean and perturbed elements. Thus, the heuristic that poisoned data will somehow be outliers in feature space does not apply. 

However, another family of defenses leverages  differential privacy (DPSGD) as a defense against poisoning \cite{ma_data_2019, hong_effectiveness_2020}. Since differentially private models are, by construction, insensitive to minor changes in the training set, they may be resistant to different poisoning methods. By clipping and noising gradients, these defenses aim to limit the effect of perturbations placed on data. In theory, this defense is applicable to our perturbations. However, we test the defense proposed in \citet{hong_effectiveness_2020} against our method and find that the drop in validation accuracy we saw in previous experiments remains even when training with DPSGD.

In addition to methods designed for defending against data poisoning, we also test the potency of our poisons under both Gaussian smoothing and random additive noise (of the same magnitude as our perturbations) during training. These modifications test how brittle our crafted perturbations are to modification. For Gaussian smoothing, we use a radius $r=2$. We find that our attack is stable under all the discussed training modifications, with none of the proposed defenses substantially improving results. These results can be found in Table \ref{tab:defenses}.

\begin{table}[thb]
\caption{Comparison of validation accuracies of a ResNet-18 with different defenses. All runs use $\varepsilon = 8/255$. }
\label{tab:defenses}
\vskip 0.15in
\begin{center}
\begin{small}
\begin{sc}
\begin{tabular}{cc}
\toprule
Defense & Validation Acc. (\%) $\downarrow$ \\
\midrule
None    & $44.82$  \\
DPSGD & $44.18$ \\
Random $\ell_\infty$ noise & $48.36$ \\
Gaussian smoothing & $24.26$ \\

\bottomrule
\end{tabular}
\end{sc}
\end{small}
\end{center}
\end{table}

\section{Facial Recognition}

\begin{table*}[htb]
\caption{Identification and verification accuracy of ResNet-18 trained on clean data and poisoned data of varying $epsilon$ radius. Note that in all cases, both identification and verification accuracy drop substantially when the model is trained on poisoned data.}
\vskip .15in
\centering

\begin{sc}
\begin{small}

\begin{tabular}{crr|rrrr}
\toprule 
& \multicolumn{2}{l}{Identification} & \multicolumn{4}{|l}{Verification}     \\
$\varepsilon$ &
  \multicolumn{1}{l}{CelebA top1} &
  \multicolumn{1}{l|}{CelebA top5} &
  \multicolumn{1}{l}{LFW} &
  \multicolumn{1}{l}{CFP} &
  \multicolumn{1}{l}{AgeDB} &
  \multicolumn{1}{l}{VGG2\_FP} \\
  \midrule 
Clean & 91.02\%          & 94.65\%         & 97.93\% & 83.84\% & 85.40\% & 84.96\% \\
8/255                    & 85.54\%          & 91.67\%         & 95.20\% & 76.40\% & 76.88\% & 81.88\% \\
16/255                   & 61.53\%          & 73.67\%         & 76.03\% & 62.81\% & 59.45\% & 61.78\% \\
32/255                   & 39.77\%          & 54.37\%         & 70.55\% & 62.94\% & 59.87\% & 59.32\% \\
\bottomrule
\end{tabular}

\end{small}
\end{sc}
\label{tab:face}
\end{table*}

\begin{table*}[htb]
\caption{Poisons generated from more a powerful model also result in a stronger attack. Here, we attack a ResNet-18 model with poisons generated with two different surrogate models at $\epsilon=8/255$}
\vskip 0.15 in
\centering
\begin{sc}
\begin{small}

\begin{tabular}{rrr|rrrr}
\toprule 
\multicolumn{1}{l}{} &
  \multicolumn{1}{l}{Identification} &
  \multicolumn{1}{l|}{} &
  \multicolumn{1}{l}{Verification} &
  \multicolumn{1}{l}{} &
  \multicolumn{1}{l}{} &
  \multicolumn{1}{l}{} \\
\multicolumn{1}{l}{Surrogate Model} &
  \multicolumn{1}{l}{CelebA top1} &
  \multicolumn{1}{l|}{CelebA top5} &
  \multicolumn{1}{l}{LFW} &
  \multicolumn{1}{l}{CFP} &
  \multicolumn{1}{l}{AgeDB} &
  \multicolumn{1}{l}{VGG2\_FP} \\
  \midrule
No Poisons & 91.02\% & 94.65\% & 97.93\% & 83.84\% & 85.40\% & 84.96\% \\
ResNet-18  & 85.54\% & 91.67\% & 95.20\% & 76.40\% & 76.88\% & 81.88\% \\
ResNet-50  & 79.53\% & 86.88\% & 86.75\% & 74.64\% & 63.78\% & 70.52\% \\
\bottomrule
\end{tabular}
\end{small}
\end{sc}
\label{tab:surrogate}
\end{table*}

While standard classification tasks like ImageNet and CIFAR allow us to establish baselines for our method, many settings where a company may wish to implement our method involve social media user data, often in the form of personal photos. Such data is may be scraped from large social media platforms by competing companies and nefarious actors \cite{cherepanova2021lowkey}. For example, companies like Clearview AI scrape photos from social media sites to train their own facial recognition systems for mass surveillance \cite{clearviewai}. A social media company could even deploy our method simply on thumbnail profile images, which are publicly available to scraping. To determine the utility of our method in preventing unauthorized use in this manner, we deploy our algorithm on facial recognition benchmarks.

By and large, facial recognition works in the regime of transfer learning. Facial recognition models are often pre-trained on a many-way classification problem using images unrelated to the testing identities. Then, during testing, the classification head is removed, and the pre-trained model is used only for calculating the test image embeddings. The identity of test images is then inferred by $k$-nearest neighbors in the embedding space. This setup is known to make attacks against facial recognition systems more difficult than attacks against standard classification tasks \cite{cherepanova2021lowkey}. Evasion attacks like the one found in \citet{cherepanova2021lowkey} modify images at test-time. We deploy our method on the complementary task of degrading the quality of the feature extractor. This adds a layer of difficulty to our attack, as the we are only able to affect the quality of the embedding, but we are not able to perturb the ``anchor" images used to classify new test samples.

\subsection{Setup}
The CelebA dataset contains 10177 identities \cite{celeba}. Following standard procedure as in \cite{facelowshot}, we remove 371 identities with too few images, use 8806 identities for pre-training and 1000 identities for testing. We use ResNet-18 and Resnet-50 as backbone architectures. During pre-training, we use the popular \textit{Cosface} classification head \cite{wang2018cosface}.  During poison dataset generation, we take 50 gradient steps with signAdam and a single random restart. In additional to the CelebA dataset, we also test the attacked model's verification accuracy on four other face datasets: Labeled Faces in the Wild (LFW) \cite{huang2008labeled}, Celebrities in Frontal-Profile data set (CPF) \cite{sengupta2016frontal},  AgeDB \cite{moschoglou2017agedb}, VGGFace2 \cite{cao2018vggface2}.  We use the \textit{face.evoLVe} repository \cite{faceevolve} to run all of our facial recognition experiments.

\subsection{Results}
In Table \ref{tab:face}, we see that both identification and verification accuracy drop materially when a model is trained on our crafted dataset. When trained on poisoned data with $\epsilon=16/255$, CelebA top 1 accuracy drops by 36\% to 61\%, . Similarly, verification accuracy also drops by as much as 35\%. Note that in many commercial applications, even a few percentage drop can tip the scales for a company to maintain a competitive advantage.

Poisons crafted on a more powerful model also transfer better in our experiments. Specifically, poisons generated with ResNet-50 always reduce both identification and verification accuracy more than poisons generated with ResNet-18 (see Table \ref{tab:surrogate}). This experiment suggests that if one wants to make the poisoned dataset more potent, one could simply use a larger capacity model to generate poisons.

\begin{figure*}[thb!]
    \centering
    \includegraphics[scale=2.8]{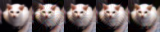}
    \caption{Example CIFAR-10 Image crafted with different regularizers. From left to right: clean image, no regularizer, $\ell_2$ regularization, SSIM regularization, TV regularization. All crafted with perturbation bound $\varepsilon = 8/255$.}
    \label{fig:regularization}
\end{figure*}

\section{Online Modification}
\label{sec:online}
The procedure outlined in Algorithm \ref{alg:crafting} works well when a practitioner already has access to a large corpus of data they intend to release. In this case, they can train a network to estimate the target gradient in Eq. \ref{eq:targ_grad}, and optimize all poisons jointly to align with this target gradient. However, for many applications, like social media release, this may not be possible. First, the initial dataset may not be big enough to allow one to train a well performing model to estimate the target gradient accurately. Second, the poisoned dataset may have to be optimized sequentially as new data continue to be released or independently at the user level.

To understand the severity of the mentioned problems, we perform several ablations and modifications to our method. First, we test whether the effective perturbations can be made with an approximated target gradient. We already estimate the target gradient on the distribution $\mathcal{D}$ by calculating the empirical target gradient on $\mathcal{S}$, but if a practitioner has access to only a small amount of data, this could degrade the quality of the target gradient to the point where perturbations become ineffective. 

To test this, we first re-run experiments outlined in Table \ref{tab:epsilon_comparison_full}, but with a target gradient approximated from a random subset of $10\%$ of the data. We find that perturbations remain effective when using this approximated target gradient. In some cases, the approximated target gradient even produces better results than the full target gradient. In theory, as long as the estimated empirical target gradient ``well" approximates the full empirical target gradient (i.e. the cosine similarity between the two is high), then perturbations crafted using the former will produce similar performance to those crafted on the latter since the target gradient is not differentiated through when crafting the perturbations. These results give confidence to practitioners who wish to poison a stream of data without having access to a large amount at the start of the process.

Additionally, we test the poisons effectiveness when the surrogate model is trained with only a small subset of the data and the target gradient is estimated only with the same small subset. This is equivalent to splitting the dataset into different subsets $\{\mathcal{S}_i\}_{i=1}^M$ where $\mathcal{S}_i \subset S$ and poisoning each of these subsets as if it were its own standalone dataset. In theory, this could harm performance since the model trained on a small subset of the data will usually perform poorly compared to a model trained on the entire dataset. This could lead to a bad target gradient estimate, and in turn, poorly performing poisons. However, even though the surrogate model's performance is indeed poorer, we find that the poisons generated with such a model are still effective in practice at decreasing validation accuracy. More specifically, in Table \ref{tab:online}, we see that when using only 10\% of data for estimating the target gradient, we can decrease the attacked model's accuracy to a comparable level as using 100\% of the data.

Finally, we test the effectiveness of the perturbations when the poisons are optimized independently as opposed to jointly. This is a practical consideration for many companies since user data could be uploaded at any given moment in time, and the practitioner might not be able to wait for a batch of data to calculate perturbations. Thus, it is necessary to be able to calculate perturbations one at a time. 

To mimic this scenario in our implementation, we simply detach the denominator in Eq. \ref{eq:alignment_loss} from the network's computation graph. Note that this change modifies the objective problem in Eq. \ref{eq:alignment_loss} to optimize the inner product of the target gradient and each \textit{individual} crafting gradient versus the cosine similarity between the target gradient and the \textit{full} crafting gradient.

 We find that poisons can be just as effective, sometimes more so, when optimized independently (see Table \ref{tab:online}). Mathematically speaking, as long as the norm of the crafting gradient in Eq. \ref{eq:alignment_loss} is stable to small changes in any individual perturbation, then detaching the denominator will not affect the perturbations as we optimize with a signed gradient method (signed Adam) in Alg. \ref{alg:crafting}. 
 
To further demonstrate this idea formally, we present the following straightforward result: 
 
 \begin{proposition}
 Fix a pixel position denoted by $*$. If $\exists \varepsilon > 0$ so that $\forall x_j \in \mathcal{S}$, in the $\ell_\infty$-ball about $x_j$ of radius $\varepsilon$, the following inequality holds:
 \begin{gather*} \bigg \vert  \frac{\partial}{\partial \Delta_j^*} \big(\Vert \nabla_\theta \mathcal{L}(\mathcal{S} + \Delta; \theta) \Vert_2 \big) \bigg \vert
 \\
<
\bigg \vert \frac{\partial}{\partial \Delta_j^*}  \bigg( \big \langle \mathcal{T}, \nabla_\theta \mathcal{L}(x_j + \Delta_j; \theta) \big \rangle \bigg ) \bigg \vert
\end{gather*}
i.e.  - The derivative w.r.t. $\Delta_j^*$ of the norm of the full crafting gradient is bounded in magnitude by the  derivative w.r.t $\Delta_j^*$ of the inner product between the individual crafting gradient and target gradient, $\mathcal{T}$,

Then, our online crafting mechanism produces the same perturbation to pixel $\Delta^*$ (in the $\varepsilon$-ball) as the full non-online version.
 \end{proposition}

 \begin{proof}
 See appendix \ref{ap:proof}.
 \end{proof}

 This adjustment makes the algorithm practical for real world use as the poisons can be generated at the user level in a distributed setting without a centralize poison generation process.



\begin{table}[]
\caption{ \textbf{\% of data used} indicates the percentage of data used for both training the surrogate model and target gradient estimation. The denominators of Eq. \ref{eq:alignment_loss} is detached when the poisons are independently crafted. We note that we can still double the validation error of the attacked model when only 5\% of data is used. }
\vskip 0.15in
\label{tab:online}
\centering
\begin{small}
\begin{sc}
\begin{tabular}{rlr}
\toprule 
\multicolumn{1}{l}{\% of data used} &
  \begin{tabular}[c]{@{}l@{}}Independently \\ Crafted\end{tabular} &
  \multicolumn{1}{l}{\begin{tabular}[c]{@{}l@{}}Poisoned\\ Validation\\ Accuracy (\%)\end{tabular}} \\
  \midrule 
100\% & No  & 46.66 \\
10\%  & No  & 44.09                     \\
5\%   & No  & 74.64                     \\
10\%  & Yes & 36.42                     \\
5\%   & Yes & 59.59                     \\

\bottomrule
\end{tabular}
\end{sc}
\end{small}
\end{table}

\label{online}

\section{Regularization}
While projecting onto an small $\ell_\infty$ ball enforces that the perturbations are not overly conspicuous, for some applications, like user uploaded images, further regularization to impose visual similarity may be desirable. Thus, we test a variety of regularization terms to the alignment loss in Eq. \ref{eq:alignment_loss}. Specifically,
we test a straightforward $\ell_2$ penalty on the norm of the perturbation, a total variation penalty (TV), and a structural similarity (SSIM) regularizer which has been shown to increase visible quality of perturbed data \cite{cherepanova2021lowkey}. We find that the regularizers decrease the visibility of the perturbations in exchange for an increase in validation accuracy. Thus, a practitioner can choose the strength of the regularizer to control the tradeoff between visibility of the perturbation and success of the poisons. Effects of the regularizers compared to an unregularized validation run can be found in Table \ref{tab:regularization}. Additionally, visualizations of images produced using the regularization terms can be found in Figure \ref{tab:regularization}.

\begin{table}[ht!]
\caption{Comparison of validation accuracy of ResNet-18 trained on poisons generated using various regularization terms. All runs use $\varepsilon = 8/255$. }
\label{tab:regularization}
\vskip 0.15in
\begin{center}
\begin{small}
\begin{sc}
\begin{tabular}{cc}
\toprule
Regularizer & Validation Acc. (\%) $\downarrow$ \\
\midrule
None    & $44.82$  \\
$\ell_2$ & $69.8$ \\
SSIM & $51.09$ \\
TV & $53.39$ \\

\bottomrule
\end{tabular}
\end{sc}
\end{small}
\end{center}
\vskip -0.1in
\end{table}

\section{Conclusions}
We develop an indiscriminate poisoning attack which allows machine learning practitioners to release data without the concern of losing the competitive advantage they gain over rival organizations by possessing access to their proprietary data. We achieve state-of-the-art results for availability attacks on modern deep networks, and we adapt our method to online modification for use by social media companies. Future areas of investigation include improving results on facial recognition tasks, possibly by attacking anchor images in addition to degrading the feature extractor. 

\section{Acknowledgements}
This work was supported by the DARPA GARD and DARPA YFA programs. Additional support was provided by DARPA QED and the National Science Foundation DMS program.
\bibliography{zotero_library,references_manual}
\bibliographystyle{icml2021}

\clearpage
\newpage
\appendix
\section{Appendix}

\subsection{Training Details}
\label{ap:training}

\textbf{CIFAR-10}
For our CIFAR-10 experiments, we train a ResNet18 model for 40 epochs in order to craft the perturbations. The model is trained with SGD and multi-step learning rate drops. However, for the experiments in Table \ref{tab:online}, since fewer data points are used for training the surrogate model, we increase the number of epochs, so the number of iteration is similar to regular training. For example, if we use only 10\% of the data, then we would increase the number of epochs by 10 times, so that the model is trained for a similar number of iterations.

Tables \ref{tab:baseline_comparison}, \ref{tab:defenses}, \ref{tab:regularization} all train a randomly initialized network using this same framework (grey-box). 

However, in Tables \ref{tab:epsilon_comparison_full}, \ref{tab:epsilon_comparison_partial} \ref{tab:online}, we use the black-box repository referenced in the main body to train models. We train for 50 epochs using this repository's setup (cosine learning rate decay, etc.). 

\textbf{ImageNet} 
For our ImageNet experiments, we use a pretrained ResNet18 model to craft poisons, and then train a new, randomly initialized ResNet18 model on these crafted poisons for 40 epochs with multi-step learning rate drops.

\textbf{Facial Recogntion} 
For our Celeb-A experiments, we train a ResNet18 and a ResNet50 for 125 epochs as surrogate models, which are used to craft poisons. We then train a new randomly initialized ResNet18 model on these crafted poisons for 125 epochs with multi-step learning rate drops. 

\textbf{DPSGD Defense}
For the DPSGD defense, we first clip employ clipping of $0.1$ and then  add noise to the gradients with parameter $\sigma = \text{clip} * 0.01$.

\subsection{Proof of Proposition}
\label{ap:proof}
\begin{proof}
Recall the alignment loss: 
\begin{equation}
    \mathcal{A}(\Delta, \theta) = 1-~\frac{\big \langle \nabla_\theta \mathcal{L}'(\mathcal{S}; \theta), \nabla_\theta \mathcal{L}(\mathcal{S} + \Delta; \theta) \big \rangle}{\Vert \nabla_\theta \mathcal{L}'(\mathcal{S}; \theta)\Vert \cdot \Vert \nabla_\theta \mathcal{L}(\mathcal{S} + \Delta; \theta) \Vert }
\end{equation}

As the left hand term in the inner product, the target gradient, does not depend on the perturbations $\Delta$, it may be treated as a constant, and denoted simply as $\mathcal{T}$. As in our algorithm, we denote $\Vert \cdot \Vert = \Vert \cdot \Vert_2$. Also, WLOG, we need only look at the sign of the following derivative for some arbitrary perturbation $\Delta_j^*$:

\[ \frac{\partial}{\partial \Delta_j^*} \bigg[ \overbrace{ 
~\frac{\big \langle \mathcal{T}, \nabla_\theta \mathcal{L}(\mathcal{S} + \Delta; \theta) \big \rangle}{\Vert \mathcal{T}\Vert \cdot \Vert \nabla_\theta \mathcal{L}(\mathcal{S} + \Delta; \theta) \Vert }}^{\alpha} 
\bigg] \]

On the one hand, if we detach the denominator from the computation graph, the following derivative is used to update the perturbation $\Delta_j$:

\[
~\frac{\frac{\partial}{\partial \Delta_j^*} \bigg( \big \langle \mathcal{T}, \nabla_\theta \mathcal{L}(x_j + \Delta_j; \theta) \big \rangle \bigg )}{\Vert \mathcal{T}\Vert \cdot \Vert \nabla_\theta \mathcal{L}(\mathcal{S} + \Delta; \theta) \Vert }  \] 

On the other hand, the gradient of the full objective is: 

\begin{gather*}
\overbrace{~\frac{\frac{\partial}{\partial \Delta_j^*} \bigg( \big \langle \mathcal{T}, \nabla_\theta \mathcal{L}(x_j + \Delta_j; \theta) \big \rangle \bigg )}{\Vert \mathcal{T}\Vert \cdot \Vert \nabla_\theta \mathcal{L}(\mathcal{S} + \Delta; \theta) \Vert }}^{\beta} 
\\
\\
- \underbrace{~\frac{ \big \langle \mathcal{T}, \nabla_\theta \mathcal{L}(x_j + \Delta_j; \theta) \big \rangle   \Vert \mathcal{T}\Vert \cdot \frac{\partial}{\partial \Delta_j^*} \bigg(\Vert \nabla_\theta \mathcal{L}(\mathcal{S} + \Delta; \theta) \Vert \bigg)}{\bigg(\Vert \mathcal{T}\Vert \cdot  \Vert \nabla_\theta \mathcal{L}(\mathcal{S} + \Delta; \theta) \Vert \bigg) ^2 } }_{\gamma} 
\end{gather*}

Note that $\beta$ is simply the derivative of the detached objective, so the perturbations will be the same if 
\[  \vert \gamma \vert < \vert \beta \vert\]

Simplifying, 
\begin{align*}
\vert \gamma \vert & = \vert \alpha \vert \cdot  ~\frac{ \bigg \vert  \frac{\partial}{\partial \Delta_j^*} \big(\Vert \nabla_\theta \mathcal{L}(\mathcal{S} + \Delta; \theta) \Vert \big) \bigg \vert}{\Vert \mathcal{T}\Vert \cdot  \Vert \nabla_\theta \mathcal{L}(\mathcal{S} + \Delta; \theta) \Vert  } 
\\
& \leq ~\frac{ \bigg \vert  \frac{\partial}{\partial \Delta_j^*} \big(\Vert \nabla_\theta \mathcal{L}(\mathcal{S} + \Delta; \theta) \Vert \big) \bigg \vert}{\Vert \mathcal{T}\Vert \cdot  \Vert \nabla_\theta \mathcal{L}(\mathcal{S} + \Delta; \theta) \Vert }
\\ 
& < 
~\frac{\bigg \vert \frac{\partial}{\partial \Delta_j^*}  \bigg( \big \langle \mathcal{T}, \nabla_\theta \mathcal{L}(x_j + \Delta_j; \theta) \big \rangle \bigg ) \bigg \vert }{\Vert \mathcal{T}\Vert \cdot \Vert \nabla_\theta \mathcal{L}(\mathcal{S} + \Delta; \theta) \Vert } = \vert \beta \vert
\end{align*}

Where the last inequality uses the assumption of the proposition, and holds in the $\varepsilon$-ball around $x_j$. Then, because our crafting algorithm uses signed gradient descent (either Adam or SGD), the perturbations crafted using the online vs. non-online method will be identical. 

\end{proof}

Note that for an alternate presentation of the proposition, we may write the bound on the magnitude of the derivative of the full crafting gradient in the following manner: 

\begin{gather*} \bigg \vert  \frac{\partial}{\partial \Delta_j^*} \big(\Vert \nabla_\theta \mathcal{L}(\mathcal{S} + \Delta; \theta) \Vert \big) \bigg \vert
 \\
<
c_0\bigg \vert \frac{\partial}{\partial \Delta_j^*} \big( \Vert \nabla_\theta \mathcal{L}(x_j + \Delta_j; \theta)\Vert \cdot \cos{\phi} \big ) \bigg \vert
\end{gather*}

Where $\phi$ is the angle between the individual crafting gradient and the fixed target gradient, and $c_0 = \Vert \mathcal{T} \Vert$ is the norm of the target gradient.

\subsection{Hardware and time considerations}
We primarily use an array of GeForce RTX 2080 Ti graphics cards. On $4$ GPUs, pre-training the crafting model on CIFAR-10 typicaly takes $12.5$ minutes. Crafting the perturbations typically takes $140$ minutes per restart for the entirety of CIFAR-10 ($50,000$ images). We typically run $240$ iterations of projected gradient descent. On ImageNet, a batch of $25,000$ Images typically takes just under $17$ hours to craft in a similar manner. 



\end{document}